\documentclass[11pt]{amsart}

\usepackage[margin=2.2cm]{geometry}
\usepackage{amsmath,amsfonts,amssymb,mathtools,bbm}
\usepackage{thmtools,thm-restate}
\usepackage{subcaption}
 \usepackage[foot]{amsaddr}

\title{The density of expected persistence diagrams and its kernel based estimation}

\author{Fr\'ed\'eric Chazal}
\address{Inria Paris-Saclay}
\email{frederic.chazal@inria.fr}
\author{Vincent Divol}
\address{Inria Paris-Saclay and Universit\'e Paris-Sud}
\email{vincent.divol@inria.fr}

\usepackage{amsthm}
\theoremstyle{plain}
\newtheorem{theorem}{Theorem}[section]

\newtheorem{proposition}[theorem]{Proposition}
\newtheorem{lemma}[theorem]{Lemma}
\newtheorem{definition}[theorem]{Definition}
\theoremstyle{remark}

\newtheorem{remark}[theorem]{Remark}

\newcommand{\R}{\mathbb{R}}

\newcommand{\X}{\mathbb{X}}

\newcommand{\FF}{\mathcal{F}}

\newcommand{\KK}{\mathcal{K}}

\newcommand{\NN}{\mathcal{N}}

\newcommand{\defeq}{\vcentcolon=}
\newcommand{\eqdef}{=\vcentcolon}

\newcommand{\dd}{\mathrm{d}}

\DeclareMathOperator*{\diam}{diam}

\newcommand{\rD}{\textbf{r}}

\newenvironment{claim}[1]{\par\noindent\underline{Claim:}\space#1}{}
\newenvironment{claimproof}[1]{\par\noindent\underline{Proof:}\space#1}{\hfill $\blacksquare$}

\begin{document}
\raggedbottom
\begin{abstract}
Persistence diagrams play a fundamental role in Topological Data Analysis where they are used as topological descriptors of filtrations built on top of data. They consist in discrete multisets of points in the plane $\R^2$ that can equivalently be seen as discrete measures in $\R^2$. When the data is assumed to be random, these discrete measures become random measures whose expectation is studied in this paper. First, we show that for a wide class of filtrations, including the \v Cech and Rips-Vietoris filtrations, but also the sublevels of a Brownian motion, the expected persistence diagram, that is a deterministic measure on $\R^2$, has a density with respect to the Lebesgue measure. Second, building on the previous result we show that the persistence surface recently introduced in \cite{adams2017persistence} can be seen as a kernel estimator of this density. We propose a cross-validation scheme for selecting an optimal bandwidth, which is proven to be a consistent procedure to estimate the density.
\end{abstract}

\maketitle

\section{Introduction}
Persistent homology (see \cite{em-ph-17} for a review), a popular approach in Topological Data Analysis (TDA), provides efficient mathematical and algorithmic tools to understand the topology of some dataset (e.g. a point cloud or a time-series) by tracking the evolution of its homology at different scales. For instance, given a scale (or time) parameter $r$ and a point cloud $x = (x_1,\dots,x_n)$ of size $n$, 
a simplicial complex $\KK(x,r)$ is built on $\{1,\dots,n\}$ thanks to some procedure, such as, e.g., the nerve of the union of balls of radius $r$ centered on the point cloud or the Vietoris-Rips complex. 
Letting the scale $r$ increase gives rise to an increasing sequence of simplicial complexes $\KK(x)=(\KK(x,r))_r$ called a {\em filtration}. When a simplex is added in the filtration at a time $r$, it either "creates" or "fills" some hole in the complex. 
Persistent homology keeps track of the birth and death of these holes and encodes them as a {\em persistence diagram} that can be seen as a relevant and stable  multi-scale topological descriptor of the data (see \cite{ccgmo-ghsssp-09,cso-psgc-13}). Similarly, one can create a filtration by considering the sublevel sets $\KK(f) = (f^{-1}(]-\infty,r]))_r$ of a given continuous real-valued function $f$ and one can track the evolution of the homology of the sublevels with very few requirements on the function $f$ (see \cite[Section 3.9]{chazal2016structure}).
A persistence diagram $D_s$ is thus a collection of pairs of numbers, each of those pairs corresponding to the birth time and the death time of a $s$-dimensional hole. A precise definition of persistence diagram can be found, for example, in \cite{em-ph-17,chazal2016structure}. Mathematically, a diagram is a multiset of points in
\begin{equation}
\Delta \defeq \{\rD = (r_1,r_2),\ r_1 <r_2 < \infty\}.
\end{equation} 
Note that in a general setting, points $\rD=(r_1,r_2)$ in diagrams can be "at infinity" on the line $\{r_2=\infty\}$ (e.g. a hole may never disappear). However, in the cases considered in this paper, this will be the case for a single point for $0$-dimensional homology, and this point will simply be discarded in the following.

In statistical settings, one is often given a (i.i.d.) sample of random datasets (either point clouds or functions in this paper) $\X_1,\dots,\X_N$ and filtrations $\KK(\X_1), \dots, \KK(\X_N)$ built on top of them. We consider the set of persistence  diagrams $D_s[\KK(\X_1)],\dots,D_s[\KK(\X_N)]$, which are thought to contain relevant topological information about the geometry of the underlying phenomenon generating the datasets. 
The space of persistence diagrams is naturally endowed with the so-called \emph{bottleneck distance} \cite{cohen2007stability} or some variants. However, the resulting metric space turns out to be highly non linear, making the statistical analysis of distributions of persistence diagrams rather awkward, despite several interesting results such as, e.g., \cite{turner2014frechet,balakrishnan2013statistical, chazal2014optimal}. A common scheme to overcome this difficulty is to create easier to handle statistics by mapping the diagrams to a vector space thanks to a feature map $\Psi$, also called a representation (see, e.g., \cite{adams2017persistence, biscio2016accumulated, bubenik2015statistical, chazal2014stochastic, chen2015statistical, kusano2016persistence, reininghaus2015stable}). A classical idea to get information about the typical behavior of an observation is then to estimate the expectation $E[\Psi(D_s[\KK(\X_i)])]$ of the distribution of representations using the mean representation 
\begin{equation}
\overline{\Psi}_N \defeq \frac{\sum_{i=1}^N \Psi(D_s[\KK(\X_i)])}{N}.
\end{equation}
In this direction, \cite{bubenik2015statistical} introduces a representation called persistence landscape, and shows 
that it satisfies law of large numbers and central limit theorems. Similar theorems can be shown for a wide variety of representations: it is known that $\overline{\Psi}_N$ is a consistent estimator of $E[\Psi(D_s[\KK(\X_i)])]$. 
Although it may be useful for a classification task, this mean representation is still somewhat disappointing from a theoretical point of view. Indeed, what exactly $E[\Psi(D_s[\KK(\X_i)])]$ is, has been scarcely studied in a non-asymptotic setting, i.e.~when the cardinality of the random point cloud $\X_i$ is fixed or bounded.

When the observed data $\X_i$s are large point clouds, asymptotic results are well understood for some non-persistent descriptors of the data, such as the Betti numbers: a natural question in geometric probability is to study the asymptotics of the $s$-dimensional Betti numbers $\beta_s(\KK(\X_n,r_n))$ where $\X_n$ is a point cloud of size $n$ and under different asymptotics for $r_n$. Notable results on the topic include \cite{kahle2013limit, yogeshwaran2015topology, Yogeshwaran2017}. Considerably less results are known about the asymptotic properties of fundamentally persistent descriptors of the data: \cite{bobrowski2017maximally} finds the right order of magnitude of maximally persistent cycles and \cite{duy2016limit} shows the convergence of persistence diagrams on stationary process in a weak sense.

\paragraph*{Contributions of the paper.} In this paper, representing persistence diagrams as discrete measures, i.e.~as element of the space of measures on $\R^2$, we establish non-asymptotic global properties of various representations and persistence-based descriptors.  A multiset of points is naturally in bijection with the discrete measure defined on $\R^2$ created by putting Dirac measures on each point of the multiset, with mass equal to the multiplicity of the point. In this paper a persistence diagram $D_s$ is thus represented as a discrete measure on $\Delta$ and with a slight abuse of notation, we will write 
\begin{equation}
D_s = \sum_{\rD \in D_s} \delta_{\rD},
\end{equation}
where $\delta_{\rD}$ denotes the Dirac measure in $\rD$ and where, as mentioned above, points with infinite persistence are simply discarded.
A wide class of representations, including the persistence surface \cite{adams2017persistence} (variants of this object have been also introduced \cite{chen2015statistical, kusano2016persistence, reininghaus2015stable}), the accumulated persistence function \cite{biscio2016accumulated} or persistence silhouette \cite{chazal2014stochastic} are conveniently expressed as $\Psi(D_s) = D_s(f) \defeq \sum_{\rD \in D_s} f(\rD)$ for some function $f$ on $\Delta$. Such representations, having particularly good theoretical properties, will be called \emph{linear representations}. Given a random set of points $\X$, the expected behavior of the linear representations $E[D_s[\KK(\X)](f)]$ is well understood if the expectation $E[D_s[\KK(\X)]]$ of the distribution of persistence diagrams is understood, where the expectation $E[\mu]$ of a random discrete measure $\mu$ is defined by the equation $E[\mu](B) = E[\mu(B)]$ for all Borel sets $B$ (see \cite{ledoux2013probability} for a precise definition of $E[\mu]$ in a more general setting). Our main contributions consists in showing that for two different kind of situations (e.g. filtrations built on point clouds in Theorem \ref{thm:main_thm} or filtration built with the sublevel sets of a Brownian motion in Theorem \ref{thm:brownian}), the expected persistence diagram $E[D_s[\KK(\X)]]$, which is a measure on $\Delta \subset \R^2$, has a density $p$ with respect to the Lebesgue measure on $\R^2$.  Therefore, $E[\Psi(D_s[\KK(\X)])]$ is equal to $\int pf$, and if properties of the density $p$ are shown (such as smoothness), those properties will also apply to the expectation of the representation $\Psi$. Note that Theorem \ref{thm:brownian} is, to our knowledge, one of the first result about the \emph{persistent} homology of Gaussian random fields. 

The main argument of the proof of Theorem \ref{thm:main_thm} relies on the basic observation that for point clouds $\X$ of given size $n$, the filtration $\KK(\X)$ can induce a finite number of ordering configurations of the simplices. The core of the proof consists in showing that, under suitable assumptions, this ordering is locally constant for almost all $\X$. As one needs to use geometric arguments, having properties only satisfied almost everywhere is not sufficient for our purpose. One needs to show that properties hold in a stronger sense, namely that the set on which it is satisfied is a dense open set. Hence, a convenient framework to obtain such properties is given by subanalytic geometry (see \cite{shiota1997geometry} for a monograph on the subject). Subanalytic sets are a class of subsets of $\R^d$ that are locally defined as linear projections of sets defined by analytic equations and inequations. As most considered filtrations in Topological Data Analysis result from real algebraic constructions, such sets naturally appear in practice. 
On open sets where the combinatorial structure of the filtration is constant, the way the points in the diagrams are matched to pairs of simplices is fixed: only the times/scales at which those simplices appear change. Under an assumption of smoothness of those times, and using the coarea formula \cite[Chapter 3]{morgan2016geometric}, a classical result of geometric measure theory generalizing the change of variables formula in integrals, one then deduces the existence of a density for $E[D_s[\KK(\X)]]$.

Among the different linear representations, persistence surface is of particular interest. It is defined as the convolution of a diagram with a gaussian kernel. Hence, the mean persistence surface can be seen as a kernel density estimator of the density $p$ of Theorem \ref{thm:main_thm}. As a consequence, the general theory of kernel density estimation applies and gives theoretical guarantees about various statistical procedures. As an illustration, we consider the bandwidth selection problem for persistence surfaces. Whereas authors in \cite{adams2017persistence} state that any reasonable bandwidth is sufficient for a classification task, we give arguments for the opposite when no "obvious" shapes appear in the diagrams. We then propose a cross-validation scheme to select the bandwidth matrix. The consistency of the procedure is shown using Stone's theorem \cite{stone1984asymptotically}. This procedure is implemented on a set of toy examples illustrating its relevance.

The paper is organized as follow: Section \ref{sec:preli} is dedicated to the necessary background in geometric measure theory and subanalytic geometry. Results are stated in Section \ref{sec:statements}, and Theorem \ref{thm:main_thm} is proved in Section \ref{sec:proof}. It is shown in Section \ref{sec:examples} that the main result applies to the \v Cech and Rips-Vietoris filtrations. Section \ref{sec:brownian} deals with the study of the persistence diagram of the Brownian motion whereas Section \ref{sec:stability} provides elements to understand the stability of the expected persistence diagrams with respect to the measure generating them. Section \ref{sec:kde} is dedicated to the statistical study of persistence surface, and numerical illustrations are found in Section \ref{sec:num}. All the technical proofs that are not essential to the understanding of the idea and results of the paper have been moved to the Appendix.

\section{Preliminaries}\label{sec:preli}
\subsection{The coarea formula}
The proof of the existence of the density of the expected persistence diagram depends heavily on a classical result in geometric measure theory, the so-called coarea formula (see \cite[Chapter 3]{morgan2016geometric} for a gentle introduction to the subject). It consists in a more general version of the change of variables formula in integrals. Let $(M,\rho)$ be a metric space. The diameter of a set $A\subset (M,\rho)$ is defined by $\sup_{x,y\in A} \rho(x,y)$.

\begin{definition} Let $k$ be a non-negative integer. For $A\subset M$, and $\delta >0$, consider 
\begin{equation}
\mathcal{H}_k^\delta(A) \defeq \inf\left\{ \sum_i \alpha(k) \left(\frac{\diam(U_i)}{2}\right)^k, \ A\subset \ \bigcup_i U_i \mbox{ and } \diam(U_i)<\delta\right\},
\end{equation}
where $\alpha(k)$ is the volume of the $k$-dimensional unit ball. The \emph{$k$-dimensional Hausdorff measure} on $M$ of $A$ is defined by $\mathcal{H}_k(A) \defeq \lim_{\delta\to 0} \mathcal{H}_k^\delta(A)$.
\end{definition}
If $M$ is a $d$-dimensional submanifold of $\R^D$, the $d$-dimensional Hausdorff measure coincides with the volume form associated to the ambient metric restricted to $M$. For instance, if $M$ is an open set of $\R^D$, the Hausdorff measure is the $D$-dimensional Lebesgue measure.

\begin{theorem}[Coarea formula \cite{morgan2016geometric}]\label{thm:coarea} Let $M$ (resp. $N$) be a smooth Riemannian manifold of dimension $m$ (resp $n$). Assume that $m\geq n$ and let $\Phi: M\to N$ be a differentiable map. Denote by $D\Phi$ the differential of $\Phi$. The Jacobian of $\Phi$ is defined by $J\Phi = \sqrt{\det((D\Phi)\times (D\Phi)^t)}$. For $f : M\to \R_+$ a positive measurable function, the following equality holds:
\begin{equation}
\int_M f(x) J\Phi(x) \dd\mathcal{H}_m(x) = \int_N \left(\int_{x\in \Phi^{-1}(\{y\})} f(x) \dd \mathcal{H}_{m-n}(x)\right) \dd\mathcal{H}_n(y).
\end{equation}
\end{theorem}

In particular, if $J\Phi>0$ almost everywhere, one can apply the coarea formula to $f\times(J\Phi)^{-1}$ to compute $\int_M f$. Having $J\Phi>0$ is equivalent to have $D\Phi$ of full rank: most of the proof of our main theorem consists in showing that this property holds for certain functions $\Phi$ of interest.

\subsection{Background on subanalytic sets}
We now give basic results on subanalytic geometry, whose proofs are given in Appendix. See \cite{shiota1997geometry} for a thorough review of the subject. Let $M\subset \R^D$ be a connected real analytic submanifold, possibly with boundary, whose dimension is denoted by $d$.

\begin{definition}
A subset $X$ of $M$ is \emph{semianalytic} if each point of $M$ has a neighbourhood $U\subset M$ such that $X \cap U$ is of the form 
\begin{equation}
\bigcup_{i=1}^p\bigcap_{j=1}^q X_{ij},
\end{equation}
where $X_{ij}$ is either $f_{ij}^{-1}(\{0\})$ or $f_{ij}^{-1}((0,\infty))$ for some analytic functions $f_{ij} : U\to \R$.
\end{definition}

\begin{definition}
A subset $X$ of $M$ is \emph{subanalytic} if for each point of $M$, there exists a neighborhood $U$ of this point, a real analytic manifold $N$ and $A$, a relatively compact semianalytic set of $N\times M$, such that $X\cap U$ is the projection of $A$ on $M$. A function $f: X\to \R$ is subanalytic if its graph is subanalytic in $M \times \R$. The set of real-valued subanalytic functions on $X$ is denoted by $\mathcal{S}(X)$.
\end{definition}
 A point $x$ in a subanalytic subset $X$ of $M$ is smooth (of dimension $k$) if, in some neighbourhood of $x$ in $M$, $X$ is an analytic submanifold (of dimension $k$). The maximal dimension of a smooth point of $X$ is called the dimension of $X$. The smooth points of $X$ of dimension $d$ are called regular, and the other points are called singular. The set $\mbox{Reg}(X)$ of regular points of $X$ is an open subset of $M$, possibly empty; the set of singular points is denoted by $\mathrm{Sing}(X)$. 

\begin{restatable}{lemma}{firstProp}
\label{lem:basic_analytic} \begin{enumerate}
\item[(i)] For $f\in \mathcal{S}(M)$, the set $A(f)$ on which $f$ is analytic is an open subanalytic set of $M$. Its complement is a subanalytic set of dimension smaller than $d$.
\end{enumerate}
Fix $X$ a subanalytic subset of $M$. Assume that $f,g: X\to \R$ are subanalytic functions such that the image of a bounded set is bounded. Then,
\begin{enumerate}
\item[(ii)] The functions $fg$ and $f+g$ are subanalytic.
\item[(iii)] The sets $f^{-1}(\{0\})$ and $f^{-1}((0,\infty))$ are subanalytic in $M$.
\end{enumerate}
\end{restatable}

As a consequence of point (i), for $f \in \mathcal{S}(M)$, one can define its gradient $\nabla f$  everywhere but on some subanalytic set of dimension smaller than $d$.
\begin{restatable}{lemma}{usefulLem}
\label{usefulLem} Let $X$ be a subanalytic subset of $M$. If the dimension of $X$ is smaller than $d$, then $\mathcal{H}_d(X)=0$.
\end{restatable}

As a direct corollary, we always have 
\begin{equation}\label{regFull}
\mathcal{H}_d(X) = \mathcal{H}_d(\mbox{Reg}(X)).
\end{equation}
Write $\NN(M)$ the class of subanalytic subsets $X$ of $M$ with $\mbox{Reg}(X)=\emptyset$. We have just shown that $\mathcal{H}_d \equiv 0$ on $\NN(M)$. They form a special class of negligeable sets. We say that a property is verified \emph{almost subanalytically everywhere} (a.s.e.) if the set on which it is not verified is included in a set of $\NN(M)$. For example, Lemma \ref{lem:basic_analytic} implies that $\nabla f$ is defined a.s.e..

\section{The density of expected persistence diagrams}\label{sec:statements}
Let $n>0$ be an integer. Write $\FF_n$ the collection of non-empty subsets of $\{1,\dots,n\}$. Let $\varphi=(\varphi[J])_{J\in \FF_n}:M^n \to \R^{\FF_n}$ be a continuous function. The function $\varphi$ will be used to construct the persistence diagram and is called a \emph{filtering function}: a simplex $J$ is added in the filtration at the time $\varphi[J]$. Write for $x=(x_1,\dots,x_n)\in M^n$ and for $J$ a simplex, $x(J)\defeq (x_j)_{j\in J}$. We make the following assumptions on $\varphi$:
\begin{enumerate}
	\item[(K1)] \emph{Absence of interaction:} For $J\in \mathcal{F}_n$, $\varphi[J](x)$ only depends on $x(J)$.
	\item[(K2)] \emph{Invariance by permutation:} For $J\in \mathcal{F}_n$ and for $(x_1,\dots,x_n)\in M^n$, if $\tau$ is a permutation of $\{1,\dots,n\}$ whose support is included in $J$, then $\varphi[J](x_{\tau(1)},\dots,x_{\tau(n)})=\varphi[J](x_1,\dots,x_n)$.
	\item[(K3)] \emph{Monotony:} For $J \subset J' \in \FF_n$, $\varphi[J] \leq \varphi[J']$.
	\item[(K4)] \emph{Compatibility:} For a simplex $J \in \FF_n$ and for $j\in J$, if $\varphi[J](x_1,\dots,x_n)$ is not a function of $x_j$ on some open set $U$ of $M^n$, then  $\varphi[J] \equiv \varphi[J\backslash\{j\}]$ on $U$.
	\item[(K5)] \emph{Smoothness:} The function $\varphi$ is subanalytic and the gradient of each of its entries (which is defined a.s.e.) is non vanishing a.s.e..
	
\end{enumerate}
Assumptions (K2) and (K3) ensure that a filtration $\KK(x)$ can be defined thanks to $\varphi$ by:
\begin{equation}
\forall J \in \mathcal{F}_n,\ J \in \KK(x,r) \Longleftrightarrow \varphi[J](x)\leq r.
\end{equation}
Assumption (K1) means that the moment a simplex is added in the filtration only depends on the position of its vertices, but not on their relative position in the point cloud. 
For $J\in \mathcal{F}_n$, the gradient of $\varphi[J]$ is a vector field in $TM^n$. Its projection on the $j$th coordinate is denoted by $\nabla^j \varphi[J]$: it is a vector field in $TM$ defined a.s.e.. The persistence diagram of the filtration $\KK(x)$ for $s$-dimensional homology is denoted by $D_s[\KK(x)]$. 

\begin{theorem}\label{thm:main_thm} Fix $n\geq 1$. Assume that $M$ is a real analytic compact $d$-dimensional connected submanifold possibly with boundary and that $\X$ is a random variable on $M^n$ having a density with respect to the Hausdorff measure $\mathcal{H}_{dn}$. Assume that $\KK$ satisfies the assumptions (K1)-(K5). Then, for $s \geq 0$, the expected measure $E[D_s[\KK(\X)]]$ has a density with respect to the Lebesgue measure on $\Delta$.
\end{theorem}

\begin{remark}The condition that $M$ is compact can be relaxed in most cases: it is only used to ensure that the subanalytic functions appearing in the proof satisfy the boundedness condition of Lemma \ref{lem:basic_analytic}. For the \v{C}ech and Rips-Vietoris filtrations, one can directly verify that the function $\varphi$ (and therefore the functions appearing in the proofs) satisfies it when $M= \R^d$. Indeed, in this case, the filtering functions are semi-algebraic.
\end{remark}

Classical filtrations such as the Rips-Vietoris and \v Cech filtrations do not satisfy the full set of assumptions (K1)-(K5). Specifically, they do not satisfy the second part of assumption (K5): all singletons $\{j\}$ are included at time $0$ in those filtrations so that $\varphi[\{j\}]\equiv 0$, and the gradient $\nabla \varphi[\{j\}]$ is therefore null everywhere. This leads to a well-known phenomenon on Rips-Vietoris and \v Cech diagrams: all the non-infinite points of the diagram for $0$-dimensional homology are included in the vertical line $\{0\}\times [0,\infty)$. A theorem similar to Theorem \ref{thm:main_thm} still holds in this case:

\begin{restatable}{theorem}{mainThmbisState}
\label{thm:main_thm'}Fix $n\geq 1$. Assume that $M$ is a real analytic compact $d$-dimensional connected submanifold and that $\X$ is a random variable on $M^n$ having a density with respect to the Hausdorff measure $\mathcal{H}_{dn}$. Define assumption (K5'):
\begin{enumerate}
\item[(K5')] The function $\varphi$ is subanalytic and the gradient of its entries $J$ of size larger than 1 is non vanishing a.s.e.. Moreover, for $\{j\}$ a singleton, $\varphi[\{j\}]\equiv 0$.
\end{enumerate}
 Assume that $\KK$ satisfies the assumptions (K1)-(K4) and (K5'). Then, for $s \geq 1$, $E[D_s[\KK(\X)]]$ has a density with respect to the Lebesgue measure on $\Delta$. Moreover, $E[D_0[\KK(\X)]]$ has a density with respect to the Lebesgue measure on the vertical line $\{0\}\times [0,\infty)$.
\end{restatable}
The proof of Theorem \ref{thm:main_thm'} is very similar to the proof of Theorem \ref{thm:main_thm}. It is therefore relegated to the appendix.

One can easily generalize Theorem \ref{thm:main_thm} and assume that the size of the point process $\X$ is itself random. For $n\in \mathbb{N}$, define a function $\varphi^{(n)} : M^n \to \R^{\FF_n}$ satisfying the assumption (K1)-(K5). If $x$ is a finite subset of $M$, define $\KK(x)$ by the filtration associated to $\varphi^{(|x|)}$ where $|x|$ is the size of $x$. We obtain the following corollary, proven in the appendix.

\begin{restatable}{corollary}{mainCorState}
\label{cor:main_cor} Assume that $\X$ has some density with respect to the law of a Poisson process on $M$ of intensity $\mathcal{H}_d$, such that $E\left[2^{|\X|}\right]<\infty$. Assume that $\KK$ satisfies the assumptions (K1)-(K5). Then, for $s \geq 0$, $E[D_s[\KK(\X)]]$ has a density with respect to the Lebesgue measure on $\Delta$.
\end{restatable}

The condition $E\left[2^{|\X|}\right]<\infty$ ensures the existence of the expected diagram and is for example satisfied when $\X$ is a Poisson process with finite intensity.

As the way the filtration is created is smooth, one may actually wonder whether the density of $E[D_s[\KK(\X)]]$ is smooth as well: it is the case as long as the way the points are sampled is smooth. Recalling that a function is said to be of class $C^k$ if it is $k$ times differentiable, with a continuous $k$th derivative, we have the following result.
\begin{restatable}{theorem}{smoothnessState}
\label{thm:smoothness} Fix $0\leq k\leq \infty$ and assume that $\X\in M^n$ has some density of class $C^k$ with respect to $\mathcal{H}_{nd}$. Then, for $s\geq 0$, the density of $E[D_s[\KK(\X)]]$ is of class $C^k$.
\end{restatable}

The proof is based on classical results of continuity under the integral sign as well as an use of the implicit function theorem: it can be found in the appendix.

As a corollary of Theorem \ref{thm:smoothness}, we obtain the smoothness of various expected descriptors computed on persistence diagrams. For instance, the expected birth distribution and the expected death distribution have smooth densities under the same hypothesis, as they are obtained by projection of the expected diagram on some axis. Another example is the smoothness of the expected Betti curves. The $s$th Betti number $\beta^r_s(\KK(x))$ of a filtration $\KK(x)$ is defined as the dimension of the $s$th homology group of $\KK(x,r)$. The Betti curves $r\mapsto \beta^r_s(\KK(x))$ are step functions which can be used as statistics, as in \cite{umeda2017time} where they are used for a classification task on time series. With few additional work (see proof in Appendix), the expected Betti curves are shown to be smooth.
\begin{restatable}{corollary}{corBettiState}
\label{cor:betti} Under the same hypothesis than Theorem \ref{thm:smoothness}, for $s\geq 0$, the expected Betti curve $ r\mapsto E[\beta^r_s(\KK(\X))]$ is a $C^k$ function.
\end{restatable}

\section{Proof of Theorem \ref{thm:main_thm}}\label{sec:proof}

First, one can always replace $M^n$ by $A(\varphi)=\bigcap_{J\in \FF_n} A(\varphi[J])$, as Lemma \ref{lem:basic_analytic} implies that it is an open set whose complement is in $\NN(M^n)$. We will therefore assume that $\varphi$ is analytic on $M^n$.

Given $x\in M^n$, the different values taken by $\varphi(x)$ on the filtration can be written $r_1 < \cdots < r_L$. Define $E_l(x)$ the set of simplices $J$ such that $\varphi[J](x) =r_l$. The sets $E_1(x),\dots,E_L(x)$ form a partition of $\mathcal{F}_n$ denoted by $\mathcal{A}(x)$. 

\begin{lemma}\label{lemma1} For a.s.e. $x\in M^n$, for $l\geq 1$,  $E_l(x)$ has a unique minimal element $J_l$ (for the partial order induced by inclusion).
\end{lemma}
\begin{proof} Fix $J,J' \subset \{1,\dots,n\}$ with $J\neq J'$ and $J\cap J'\neq \emptyset$. consider the subanalytic functions $f: x\in M^n \mapsto \varphi[J](x)-\varphi[J'](x)$ and $g: x\in M^n \mapsto \varphi[J](x)-\varphi[J\cap J'](x)$. The set
\begin{equation}
C(J,J') \defeq \{ f=0 \} \cap \{ g>0\}.
\end{equation}
is a subanalytic subset of $M^n$. Assume that it contains some open set $U$. On $U$, $\varphi[J](x)$ is equal to $\varphi[J'](x)$. Therefore, it does not depend on the entries $x_j$ for $j \in J \backslash J'$. Hence, by assumption (K4), $\varphi[J](x)$ is actually equal to $\varphi[J\cap J'](x)$ on $U$. This is a contradiction with having $g>0$ on $U$. Therefore, $C(J,J')$ does not contain any open set, and all its points are singular: $C(J,J')$ is in $\mathcal{N}(M^n)$. If $J\cap J'=\emptyset$, similar arguments show that $C(J,J') = \{f=0\}$ cannot contain any open set: it would contradict assumption (K5). On the complement of 
\begin{equation}
C\defeq \bigcup_{J\neq J' \subset\{1,\dots,n\}} C(J,J'),
\end{equation}
having $\varphi[J](x) = \varphi[J'](x)$ implies that this quantity is equal to $\varphi[J\cap J'](x)$. This show the existence of a unique minimal element $J_l$ to $E_l(x)$ on the complement of $C$. This property is therefore a.s.e. satisfied. 
\end{proof}

\begin{lemma}\label{lemma2} A.s.e., $x \mapsto \mathcal{A}(x)$ is locally constant.
\end{lemma}

\begin{proof} Fix $\mathcal{A}_0 = \{E_1,\dots,E_l\}$ a partition of $\mathcal{F}_n$ induced by some filtration, with minimal elements $J_1,\dots,J_l$. 
Consider the subanalytic functions $F, G$ defined, for $x\in M^n$, by 
\[ F(x) = \sum_{l=1}^L \sum_{J\in E_l} \left( \varphi[J](x)-\varphi[J_l](x)\right) \mbox{  and  } G(x) = \sum_{l\neq l'} \left(\varphi[J_l](x))-\varphi[J_{l'}](x)\right)^2.\]
The set $\{x \in M^n, \mathcal{A}(x)=\mathcal{A}_0\}$ is exactly the set $C(\mathcal{A}_0)=\{F=0\} \cap \{G>0\}$, which is subanalytic. The sets $C(\mathcal{A}_0)$ for all partitions $\mathcal{A}_0$ of $\mathcal{F}_n$ define a finite partition of the space $M^n$. On each open set $\mbox{Reg}(C(\mathcal{A}_0)))$, the application $x\mapsto \mathcal{A}(x)$ is constant. Therefore, $x\mapsto \mathcal{A}(x)$ is locally constant everywhere but on $\bigcup_{\mathcal{A}_0} \mathrm{Sing}(C(\mathcal{A}_0)) \in \NN(M^n)$.
\end{proof}

Therefore, the space $M^n$ is partitioned into a negligeable set of $\NN(M^n)$ and some open subanalytic sets $U_1,\dots,U_R$ on which $\mathcal{A}$ is constant.

\begin{lemma}\label{lem:grad_non_null} Fix $1\leq r \leq R$ and assume that $J_1,\dots,J_L$ are the minimal elements of $\mathcal{A}$ on $U_r$. Then, for $1\leq l \leq L$ and $j\in J_l$, $\nabla^j \varphi[J_l] \neq 0$ a.s.e. on $U_r$.
\end{lemma}
\begin{proof} By minimality of $J_l$, for $j\in J_l$, the subanalytic set $\{\nabla^j \varphi[J_l] = 0\} \cap U_r$ cannot contain an open set. It is therefore in $\NN(M^n)$.
\end{proof}

Fix $1\leq r \leq R$ and write 
\[V_r = U_r\ \Big\backslash \left(\bigcup_{l=1}^L \bigcup_{j=1}^{|J_l|} \{ \nabla^j \varphi[J_l] = 0\}\right).\]
 The complement of $V_r$ in $U_r$ is still in $\NN(M^n)$. For $x\in V_r$, $D_s[\KK(x)]$ is written $\sum_{i=1}^N \delta_{\rD_i}$, where 
 \[\rD_i = (\varphi[J_{l_1}](x),\varphi[J_{l_2}](x))\eqdef (b_i,d_i).\] The integer $N$ and the simplices $J_{l_1}$, $J_{l_2}$ depend only on $V_r$. Note that $d_i$ is always larger than $b_i$, so that $J_{l_2}$ cannot be included in $J_{l_1}$. The map $x\mapsto \rD_i$ has it differential of rank 2. Indeed, take $j \in J_{l_2} \backslash J_{l_1}$. By Lemma \ref{lem:grad_non_null}, $\nabla^j \varphi[J_{l_2}](x) \neq 0$. Also, as $\varphi[J_{l_1}]$ only depends on the entries of $x$ indexed by $J_{l_1}$ (assumption (K1)), $\nabla^j \varphi[J_{l_1}](x)=0$. Furthermore, take $j'$ in $J_{l_1}$. By Lemma \ref{lem:grad_non_null}, $\nabla^{j'} \varphi[J_{l_1}](x)\neq 0$. This implies that the differential is of rank 2.

We now compute the $s$th persistence diagram for $s \geq 0$. Write $\kappa$ the density of $\X$ with respect to the measure $\mathcal{H}_{nd}$ on $M^n$. Then,
 \begin{align*}
 E[D_s[\KK(\X)]] &= \sum_{r=1}^R E\left[ \mathbbm{1}\{\X \in V_r\} D_s[\KK(\X)] \right] =\sum_{r=1}^R E\left[ \mathbbm{1}\{\X \in V_r\} \sum_{i=1}^{N_r} \delta_{\rD_i} \right]\\
 &= \sum_{r=1}^R  \sum_{i=1}^{N_r} E\left[ \mathbbm{1}\{\X \in V_r\}\delta_{\rD_i} \right]
 \end{align*}
Write $\mu_{ir}$ the measure $E[\mathbbm{1}\{\X \in V_r\} \delta_{\rD_i}]$. To conclude, it suffices to show that this measure has a density with respect to the Lebesgue measure on $\Delta$. This is a consequence of the coarea formula. Define the function $\Phi_{ir} : x\in V_r \mapsto \rD_i = (\varphi[J_{l_1}](x),\varphi[J_{l_2}](x))$. We have already seen that $\Phi_{ir}$ is of rank $2$ on $V_r$, so that $J\Phi_{ir}>0$. By the coarea formula (see Theorem \ref{thm:coarea}), for a Borel set $B$ in $\Delta$,
\begin{align*}
\mu_{ir}(B) = P(\Phi_{ir}(\X) \in B, \X \in V_r) &= \int_{V_r} \mathbbm{1}\{\Phi_{ir}(x) \in B\}\kappa(x)d\mathcal{H}_{nd}(x) \\
&=  \int_{u\in B} \int_{x\in \Phi_{ir}^{-1}(\{u\})} (J\Phi_{ir}(x))^{-1}\kappa(x) d\mathcal{H}_{nd-2}(x) du.
\end{align*}
Therefore, $\mu_{ir}$ has a density with respect to the Lebesgue measure on $\Delta$ equal to
\begin{equation}\label{densityLoc}
p_{ir}(u) = \int_{x\in \Phi_{ir}^{-1}(\{u\})} (J\Phi_{ir}(x))^{-1}\kappa(x) d\mathcal{H}_{nd-2}(x).
\end{equation}
Finally, $E[D_s[\KK(\X)]]$ has a density equal to
\begin{equation}\label{fullDens}
p(u) = \sum_{r=1}^R \sum_{i=1}^{N_r}  \int_{x\in \Phi_{ir}^{-1}(\{u\})} (J\Phi_{ir}(x))^{-1}\kappa(x) d\mathcal{H}_{nd-2}(x).
\end{equation}

\begin{remark}
Notice that, for $n$ fixed, the above proof, and thus the conclusion, of Theorem \ref{thm:main_thm} also works if the diagrams are represented by normalized discrete measures, i.e.~probability measures defined by
\begin{equation}
D_s = \frac{1}{|D_s|} \sum_{\rD \in D_s} \delta_{\rD}.
\end{equation}
\end{remark}

\section{Examples}\label{sec:examples}
We now note that the Rips-Vietoris and the \v{C}ech filtrations satisfy the assumptions (K1)-(K4) and (K5') when $M = \R^d$ is an Euclidean space. Note that the similar arguments show that weighted versions of those filtrations (see \cite{buchet2016efficient}) satisfy assumptions (K1)-(K5).

\subsection{Rips-Vietoris filtration}
For the Rips-Vietoris filtration, $\varphi[J](x) = \max_{i,j\in J} \|x_i-x_j\|$. The function $\varphi$ clearly satisfies (K1), (K2) and (K3). It is also subanalytic, as it is the maximum of semi-algebraic functions.

Let $x\in M^n$ and $J\in \mathcal{F}_n$ a simplex of size larger than one. Then, $\varphi[J](x)=\|x_i-x_j\|$ for some indices $i,j$. Those indices are locally stable, and $\varphi[J](x)=\varphi[\{i,j\}](x)$: hypothesis (K4) is satisfied. Furthermore, on this set, 
\begin{equation}
\nabla \varphi[\{i,j\}](x) = \left(\frac{x_i-x_j}{\|x_i-x_j\|},\frac{x_j-x_i}{\|x_i-x_j\|}\right) \neq 0.
\end{equation} 
Hence, (K5') is also satisfied: both Theorem \ref{thm:main_thm'} and Theorem \ref{thm:smoothness} are satisfied for the Rips-Vietoris filtration.

\subsection{\v Cech filtration}
The ball centered at $x$ of radius $r$ is denoted by $B(x,r)$. For the \v{C}ech filtration, 
\begin{equation}
 \varphi[J](x) = \inf_{r>0} \left\{ \bigcap_{j\in J} B(x_j,r) \neq \emptyset \right\}.
\end{equation}
First, it is clear that (K1), (K2) and (K3) are satisfied by $\varphi$.

We give without proof a characterization of the \v{C}ech complex.
\begin{proposition} Let $x$ be in $M^n$ and fix $J\in \mathcal{F}_n$. If the circumcenter of $x(J)$ is in the convex hull of $x(J)$, then $\varphi[J](x)$ is the radius of the circumsphere of $x(J)$. Otherwise, its projection on the convex hull belongs to the convex hull of some subsimplex $x(J')$ of $x(J)$ and $\varphi[J](x)=\varphi[J'](x)$.
\end{proposition}

\begin{definition} The Cayley-Menger matrix of a $k$-simplex $x=(x_1,\dots,x_k)\in M^k$ is the symmetric matrix $(M(x)_{i,j})_{i,j}$ of size $k+1$, with zeros on the diagonal, such that $M(x)_{1,j}=1$ for $j>1$ and $M(x)_{i+1,j+1} = \|x_i-x_j\|^2$ for $i,j\leq k$.
\end{definition}

\begin{proposition}[see \cite{coxeter1930circumradius}]Let $x \in M^k$ be a point in general position. Then, the Cayley-Menger matrix $M(x)$ is invertible with $(M(x))^{-1}_{1,1} = -2r^2$, where $r$ is the radius of the circumsphere of $x$. The $k$th other entries of the first line of $M(x)^{-1}$ are the barycentric coordinates of the circumcenter.
\end{proposition}

Therefore, the application which maps a simplex to its circumcenter is analytic, and the set on which the circumcenter of a simplex belongs in the interior of its convex hull is a subanalytic set. On such a set, the function $\varphi$ is also analytic, as it is the square root of the inverse a matrix which is polynomial in $x$. Furthermore, on the open set on which the circumcenter is outside the convex hull, we have shown that $\varphi[J](x)=\varphi[J'](x)$ for some subsimplex $J'$: assumption (K4) is satisfied.

Finally, let us show that assumption (K5') is satisfied. The previous paragraph shows the subanalyticity of $\varphi$. For $J\in \mathcal{F}_n$ a simplex of size larger than one, there exists some subsimplex $J'$ such that $\varphi[J](x)$ is the radius of the circumsphere of $x(J')$. It is clear that there cannot be an open set on which this radius is constant. Thus, $\nabla \varphi[J]$ is a.s.e. non null.

\section{The expected persistence diagram of a Brownian motion}\label{sec:brownian}
Another instance of random objects one can build filtrations on are random functions. The most fundamental instance of such functions is the Brownian motion $B: t\in[0,1] \mapsto B_t\in \R$, defined as the continuous Gaussian random field on $\R$ having covariance function $C(t_1,t_2)=\min(t_1,t_2)$ (see \cite[Chapter 2]{le2016brownian} for a concise and rigorous introduction). The continuity of $B$ ensures that the persistence module induced by the $0$-level homology of its sublevel sets is \emph{q-tame} \cite[Section 3.9]{chazal2016structure}. In particular, the persistence diagram $D$ of this persistence module is well-defined, but may contain accumulation points close to the diagonal. From a measure point of view, the persistence diagram is not a finite measure as in previous sections, but a Radon measure on $\Delta$. 

\begin{theorem}\label{thm:brownian} The random persistence diagram $D$ of the $0$-level homology of the sublevel sets of $B$ is such that its expectation $E[D]$ is well defined and has a density with respect to the Lebesgue measure.
\end{theorem}

The result holds as the persistent Betti numbers, defined by \[\beta^{r,s} \defeq D(]-\infty,r] \times [s,\infty[),\] have a particularly convenient expression in this setting. Indeed, $\beta^{r,s}$ is exactly the number of upward crossings of the band $[r,s]$ by the Brownian motion. The law of this quantity is explicitly known, and happens to be continuous with respect to $r$ and $s$. Standard measure theoretic arguments are then enough to conclude.

\begin{figure}[h]
    \centering
    \includegraphics[width=0.8\textwidth]{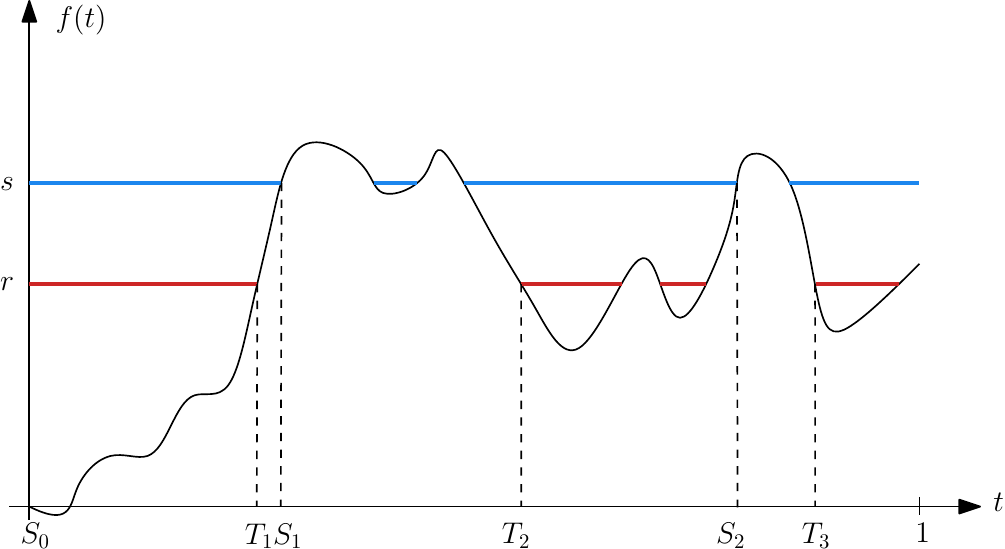}
    \caption{Example of a function $f:[0,1] \to \R$ with $\beta^{r,s}= 3$. The red region corresponds to $f^{-1}((-\infty,r])$ and the blue region to $f^{-1}((-\infty,s])$.}\label{fig:brownian}
  \end{figure}

More precisely, for $a\in \R$, define $T(a) \defeq \inf \{t>0,\ B_t=a\}$. Then, $T(a)$ has a density with respect to the Lebesgue measure equal to
\[ f_a(t) = \frac{a}{\sqrt{2\pi t^3}} \exp\left( -\frac{a^2}{2t}\right).\]
Assume that $0<r<s$ (similar arguments hold when both numbers are negative or if $r<0<s$). Define $T_0=S_0 = 0$ and for $i\geq 0$
\begin{align*}
 &T_{i+1} \defeq \inf \{ t\geq S_i,\ B_t = r\},\\
& S_{i+1} \defeq \inf\{ t \geq T_{i+1}, \ B_t =s\}.
\end{align*}
Then $\beta^{r,s}$ is equal to $\max\{ k\geq 0,\ T_k \leq 1\}$ (see Figure \ref{fig:brownian})  and $P(\beta^{r,s}\geq k)  = P(T_k \leq 1)$. First, note that $T_1$ is equal to $T(r)$. Also, by Markov property, for $i\geq 1$, conditionally on $S_i$, $T_{i+1}-S_i$ has the same law than $T(s-r)$, and so does $S_{i+1}-T_{i+1}$ conditionally on $T_{i+1}$. Therefore for $k\geq 2$,
\begin{align*}
P(\beta^{r,s}\geq k) &= P(T_k \leq 1)  = \int_{\Sigma_{2k-2}} f_r(t_1)f_{s-r}(s_1)f_{s-r}(t_2)\cdots f_{s-r}(s_{k-1})f_{s-r}(t_k) \dd s \dd t,
 \end{align*}
where $\Sigma_{2k-2} \defeq \{ u=(t_1,\dots,t_k,s_1,\dots,s_{k-1}) \in \R^{2k-1}, \ t_i\geq 0,\ s_i\geq 0 \mbox{ and } \sum_{i=1}^k t_i + \sum_{i=1}^{k-1} s_i \leq 1\}$ is the unit simplex of dimension $2k-2$.
 Therefore,
 \begin{align*} E[\beta^{r,s}] &= \sum_{k \geq 1}  P(\beta^{r,s} \geq k) = \sum_{k \geq 1} \int_{\Sigma_{2k-2}} f_r(t_1)f_{s-r}(s_1)f_{s-r}(t_2)\cdots f_{s-r}(s_{k-1})f_{s-r}(t_k) \dd s \dd t  \\
 &= \sum_{k\geq 1} \int_{\Sigma_{2k-2}} \frac{r(s-r)^{2k-2}}{\prod_{i=1}^{2k-1} \sqrt{2 \pi  u_i^3}} \exp\left( -\frac{r^2}{2u_1} - \frac{(s-r)^2}{2}\sum_{i=2}^{2k-1} u_i^{-1} \right) \dd u  \\
 &\defeq \sum_{k\geq 1} \int_{\Sigma_k} G_k(u;r,s)\dd u  \defeq \sum_{k \geq 1} I_k(r,s) .
 \end{align*}

Note first that this sum is finite. Indeed, for $b \geq 0$, the function $x\in [0,1] \mapsto x^{-1} + \frac{\ln(x)}{b}$ is bounded from below by $b^{-1}(1+\ln(b))$. Therefore,
\begin{align*} 
G_k(u;r,s) &\leq \frac{r(s-r)^{2k-2}}{(2\pi)^{k-1/2}} \exp \bigg( -\frac{r^2}{2} \frac{3}{r^2} \left(1 + \ln\left( \frac{r^2}{3} \right)\right) \\
&\hspace{2cm} -(2k-2)\frac{(s-r)^2}{2} \frac{3}{(s-r)^2} \left(1 + \ln\left( \frac{(s-r)^2}{3} \right)\right) \bigg)\\
&= \frac{r(s-r)^{2k-2}}{(2\pi)^{k-1/2}} \exp \left( -(2k-1) \frac{3}{2}(1-\ln 3) - \ln(r^3)-  (2k-2)\ln ((s-r)^3) \right) \\
&=\frac{r^{-2}(s-r)^{-4(k-1)}}{(2\pi)^{k-1/2}} BC^k
\end{align*}
for some constants $B,C$.
As the volume of $\Sigma_k$ is $\frac{\sqrt{k+1}}{k!}$, $\sum_{k\geq 0} I_k(r,s)$ is finite. Moreover, it is possible to find a local bound of $I_k(r,s)$ independent of $r$ and $s$: using classical results on the continuity of parametric integrals, one has that $E[\beta^{r,s}] = \mu(A_{r,s})$ is continuous in $r$ and $s$. Using the similar bounds on the derivatives of $I_k(r,s)$, one can show that $(r,s) \mapsto \mu(A_{r,s})$ is a $C^1$ function. This implies that $\mu$ is absolutely continuous with respect to the Lebesgue measure on $\Delta$.

\section{A few remarks on the stability of expected persistence diagrams}\label{sec:stability}
In two different situations, namely for point clouds in Section \ref{sec:statements} and for sublevels of functions in Section \ref{sec:brownian}, we have described how to define a map $P \mapsto E[D(\KK(\X))]$ where $\X$ has distribution $P$ and $P$ is a probability distribution on either $M^n$ or $C([0,1])$, the space of continuous functions defined on $[0,1]$. The continuity (or even Lipschitz-continuity) of such a map with respect to some metrics is a natural question, having both theoretical and practical implications: in particular, it implies the stability of mean linear representations, which justifies the use of such representations to perform statistical inference. We propose partial answers to this general question, with metrics measured with $L_1$ and $L_\infty$ distances.

\begin{theorem} Let $n\geq 1$ and $M$ be a real analytic compact $d$-dimensional connected submanifold. Let $\X_1$ (resp. $\X_2$) be a random variable on $M^n$ having a density $\kappa_1$ (resp. $\kappa_2$) with respect to the Hausdorff measure $\mathcal{H}_{dn}$. Assume that $\KK$ satisfies the assumptions (K1)-(K5) (or (K5')). Let $\overline{p}_1$ be the density of the \emph{normalized} measure $E\left[\frac{D_s[\KK(\X_1)]}{|D_s[\KK(\X_1)]|}\right]$ and $\overline{p}_2$ be the density of $E\left[\frac{D_s[\KK(\X_2)]}{|D_s[\KK(\X_2)]|}\right]$. Also, let $p_1$ and $p_2$ be the non-normalized densities.Then,
\begin{align}
\|\overline{p}_1-\overline{p}_2\|_1 &\leq \|\kappa_1-\kappa_2\|_1, \text{ and }\label{eq:bound_normalized}\\
\|p_1-p_2\|_1 &\leq C_n\mathcal{H}_d(M)^n\|\kappa_1-\kappa_2\|_\infty, \label{eq:bound_infty}
\end{align}
where $C_n$ is the expected number of points in the persistence diagram built with the filtration $\KK$ on $n$ i.i.d. uniform points on $M$.
\end{theorem}
It is conjectured (and even proved for $M=[0,1]^d$ in a parallel work \cite{divol2018choice}) that $C_n$ is of order $n$ when $\KK$ is either the Rips or the \v Cech filtration.

\begin{proof}
Consider first the non-normalized case. Given the expression \eqref{densityLoc}, one can write for $u\in \Delta$:
\begin{align*}
p_1(u) -p_2( u) &= \sum_{r=1}^R \sum_{i=1}^{N_r}  \int_{x\in \Phi_{ir}^{-1}(\{u\})} (J\Phi_{ir}(x))^{-1}(\kappa_1(x)-\kappa_2(x)) d\mathcal{H}_{nd-2}(x) \\
\int_\Delta |p_1(u)-p_2(u)|du &\leq  \sum_{r=1}^R \sum_{i=1}^{N_r}  \int_\Delta\int_{x\in \Phi_{ir}^{-1}(\{u\})} (J\Phi_{ir}(x))^{-1}|\kappa_1(x)-\kappa_2(x)| d\mathcal{H}_{nd-2}(x) \\
&= \sum_{r=1}^R \sum_{i=1}^{N_r}  \int_{V_r} \mathbbm{1}\{\Phi_{ir}(x) \in \Delta \} |\kappa_1(x)-\kappa_2(x)| d\mathcal{H}_{nd}(x) \text{ by the coarea formula} \\
&= \sum_{r=1}^R N_r \int_{V_r}  |\kappa_1(x)-\kappa_2(x)| d\mathcal{H}_{nd}(x) \\
&\leq \sum_{r=1}^R N_r \mathcal{H}_{nd}(V_r) \|\kappa_1-\kappa_2\|_\infty  \\
&= \mathcal{H}_{nd}(M^n) \sum_{r=1}^R N_r \frac{\mathcal{H}_{nd}(V_r)}{\mathcal{H}_{nd}(M^n)} \|\kappa_1-\kappa_2\|_\infty = \mathcal{H}_d(M)^n C_n  \| \kappa_1 - \kappa_2\|_\infty.
\end{align*}
Inequality \eqref{eq:bound_normalized} is likewise obtained.
\end{proof}

\begin{remark} Other metrics of interest on the space of persistence diagrams are Wasserstein metrics $d_p$, defined as the minimal cost of some matchings over the points of two diagrams. Endowed with those metrics, persistence diagrams are known to satisfy strong stability results with respect to the data they are built with (see \cite{cohen2010lipschitz}). It would therefore be expected that a similar stability holds for the expectation of random diagrams. However, the expected diagrams are not persistence diagrams, but Radon measures on $\Delta$. It is therefore first needed to extend $d_p$ to this more general space in a meaningful way. Similar technique to the one used in \cite{chazal2015subsampling} would then be sufficient to conclude. Extending the $d_p$ measures to Radon measures is the topic of a parallel work, see \cite{divol2019understanding}.
\end{remark}

\section{Persistence surface as a kernel density estimator}\label{sec:kde}
Persistence surface is a representation of persistence diagrams introduced by \cite{adams2017persistence}. It consists in a convolution of a diagram with a kernel, a general idea that has been repeatedly and fruitfully exploited, with slight variations, for instance in \cite{chen2015statistical, kusano2016persistence, reininghaus2015stable}.  For $K:\R^2\to \R$ a kernel and $H$ a bandwidth matrix (e.g. a symmetric positive definite matrix), let for $u\in \R^2$,
\begin{equation}
K_H(u) = \det(H)^{-1/2} K(H^{-1/2}\cdot u).
\end{equation}
For $D$ a diagram, $K : \R^2 \to \R$ a kernel, $H$ a bandwidth matrix and $w:\R^2 \to \R_+$ a weight function, one defines the persistence surface of $D$ with kernel $K$ and weight function $w$ by:
\begin{equation}
\forall u \in \R^2, \ \rho(D)(u) \defeq \sum_{\rD \in D} w(\rD)K_H(u-\rD) = D(wK_H(u-\cdot))
\end{equation}

Assume that $\X$ is some point process satisfying the assumptions of Theorem \ref{thm:main_thm}. Then, for $s\geq 1$, $\mu \defeq E[D_s[\KK(\X)]]$ has some density $p$ with respect to the Lebesgue measure on $\Delta$. Therefore, $\mu_w$, the measure having density $w$ with respect to $\mu$, has a density equal to $w\times p$ with respect to the Lebesgue measure. The mean persistence surface $E[\rho(D_s[\KK(\X)])]$ is exactly the convolution of $\mu_w$ by some kernel function: the persistence surface $\rho(D_s[\KK(\X)])$ is actually a kernel density estimator of $w\times p$.

If a point cloud approximates a shape, then its persistence diagram (for the \v{C}ech filtration for instance) is made of numerous points with small persistences and a few meaningful points of high persistences which corresponds to the persistence diagram of the "true" shape. As one is interested in the latter points, a weight function $w$, which is typically an increasing function of the persistence, is used to suppress the importance of the topological noise in the persistence surface. \cite{adams2017persistence} argue that in this setting, the choice of the bandwidth matrix $H$ has few effects for statistical purposes (e.g. classification), a claim supported by numerical experiments on simple sets of synthetic data, e.g. torus, sphere, three clusters, etc.

However, in the setting where the datasets are more complicated and contain no obvious "real" shapes, one may expect the choice of the bandwidth parameter $H$ to become more critical: there are no highly persistent, easily distinguishable points in the diagrams anymore and the precise structure of the density functions of the processes becomes of interest. We show that a cross validation approach allows the bandwidth selection task to be done in an asymptotically consistent way.
This is a consequence of a generalization of Stone's theorem \cite{stone1984asymptotically} when observations are not random vectors but random measures.

Assume that $\mu_1,\dots,\mu_N$ are i.i.d. random measures on $\R^2$, such that there exists a deterministic constant $C$ with $|\mu_1|\leq C$. Assume that the expected measure $E[\mu_1]$ has a bounded density $p$ with respect to the Lebesgue measure on $\R^2$. Given a kernel $K:\R^2\to \R$ and a bandwidth matrix $H$, one defines the kernel density estimator
\begin{equation}
\hat{p}_{H}(x) \defeq \frac{1}{N} \sum_{i=1}^N \int K_H(x-y)\mu_i(dy).
\end{equation}

The optimal bandwidth $H_{opt}$ minimizes the Mean Integrated Square Error (MISE)
\begin{equation}
MISE(H) \defeq E\left[\|p-\hat{p}_H\|^2\right] = E\left[ \int \left(p(x)-\hat{p}_H(x)\right)^2 dx\right].
\end{equation}
Of course, as $p$ is unknown, $MISE(H)$ cannot be computed. Minimizing $MISE(H)$ is equivalent to minimize $J(H) \defeq MISE(H)-\|p\|^2$. Define
\begin{equation}
\hat{p}_{iH}(x) \defeq \frac{1}{N-1} \sum_{j\neq i} \int K_H(x-y)\mu_j(dy)
\end{equation}
and
\begin{equation}\label{crit}
\hat{J}(H) \defeq \frac{1}{N^2} \sum_{i,j} \iint K_H^{(2)}(x-y)\mu_i(dx)\mu_j(dy) -\frac{2}{N} \sum_i \int \hat{p}_{iH}(x)\mu_i(dx),
\end{equation}
where $K^{(2)}: x \mapsto \int K(x-y)K(y)dy$ denotes the convolution of $K$ with itself. The quantity $\hat{J}(H)$ is an unbiased estimator of $J(H)$. The selected bandwidth $\hat{H}$ is then chosen to be equal to $\arg \min_H \hat{J}(H)$.

\begin{theorem}[Stone's theorem \cite{stone1984asymptotically}]\label{thm:stone}  Assume that the kernel $K$ is nonnegative, H\"older continuous and has a maximum attained in $0$. Also, assume that the density $p$ is bounded. Then, $\hat{H}$ is asymptotically optimal in the sense that
\begin{equation}
\frac{\|p-\hat{p}_{\hat{H}}\|}{\|p-\hat{p}_{H_{opt}}\|} \xrightarrow[N\to \infty]{} 1 \mbox{ a.s..}
\end{equation}
\end{theorem}
Note that the gaussian kernel $K(x) = \exp(-\|x\|^2/2)$ satisfies the assumptions of Theorem \ref{thm:stone}. 

The quality of the optimal estimator can also be studied. Indeed, a straightforward adaptation of the classical study of kernel density estimator (as presented for example in \cite{tsybakov2008introduction}) to the case of a sample of i.i.d. random measures shows that there exists a choice $H_N$ of bandwidth depending on $N$ and on the (unknown) regularity of $p$ such that the $\hat{p}_{H_N}$ is a consistent estimator of $p$ in the sense that $E[\|p-\hat{p}_{H_N}\|^2] \to 0$ (with known rate of convergence).  Therefore, Theorem \ref{thm:stone} asserts that the cross-validation procedure is consistent.

Let $\X_1,\dots,\X_N$ be i.i.d. processes on $M$ having a density with respect to the law of a Poisson process of intensity $\mathcal{H}_d$. Assume that there exists a deterministic constant $C$ with $|\X_i|\leq C$. Then, Theorem \ref{thm:stone} can be applied to $\mu_i = D_s[\KK(\X_i)]$. Therefore, \emph{the cross validation procedure \eqref{crit} to select $H$ the bandwidth matrix in the persistence surface ensures that the mean persistence surface}
\begin{equation}
\overline{\rho}_N \defeq \frac{1}{N} \sum_{i=1}^N \rho(D_s[\KK(\X_i)])
\end{equation}
\emph{is a consistent estimator of $p$ the density of $E[D_s[\KK(\X_1)]]$.} 

\section{Numerical illustration}\label{sec:num}

Three sets of synthetic data are considered (see Figure \ref{fig:datasets}). The first one (a) is made of $N=40$ sets of $n = 300$ i.i.d. points uniformly sampled in the square $[0,1]^2$. The second one (b) is made of $N$ samples of a clustered process: $n/3$ cluster's centers are uniformly sampled in the square. Each center is then replaced with $3$ i.i.d. points following a normal distribution of standard deviation $0.01\times n^{-1/2}$. The third dataset (c) is made of $N$ samples of $n$ uniform points on a torus of inner radius $1$ and outer radius $2$. For each set, a \v Cech persistence diagram for $1$-dimensional homology is computed. Persistence diagrams are then transformed under the map $(r_1,r_2) \mapsto (r_1,r_2-r_1)$, so that they now live in the upper-left quadrant of the plane. Figure \ref{fig:diagrams} shows the superposition of the diagrams in each class. One may observe the slight differences in the structure of the topological noise over the classes (a) and (b). The cluster of most persistent points in the diagrams of class (c) correspond to the two holes of a torus and are distinguishable from the rest of the points in the diagrams of the class, which form topological noise. The persistence diagrams are weighted by the weight function $w(\rD)=(r_2-r_1)^3$, as advised in \cite{kusano2017kernel} for two-dimensional point clouds. The bandwidth selection procedure will be applied to the measures having density $w$ with respect to the diagrams, e.g. a measure is a sum of weighted Dirac measures.

\begin{figure}
    \centering
    \begin{subfigure}[b]{0.3\textwidth}
        \includegraphics[width=\textwidth]{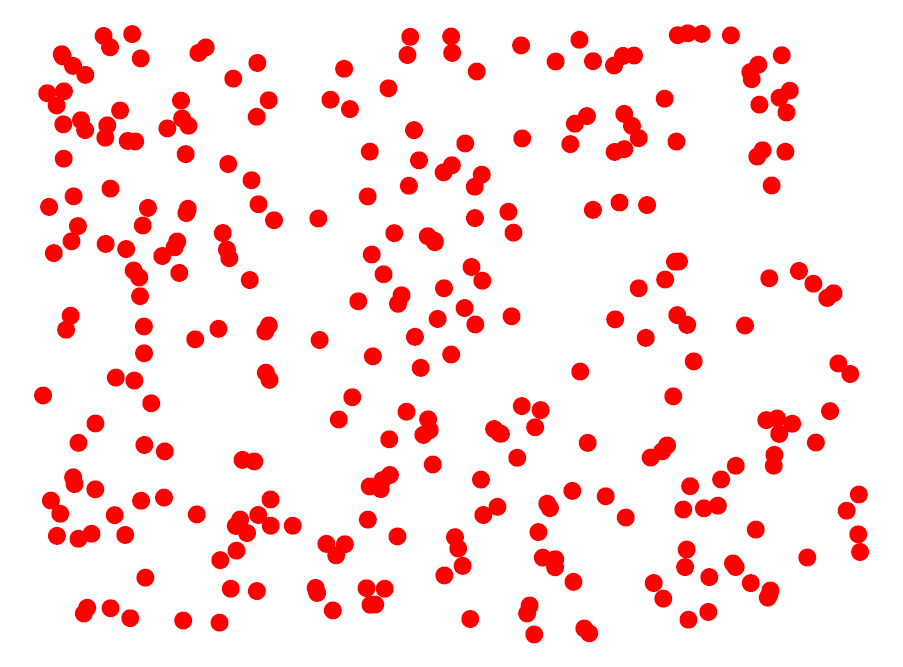}
        \caption{}
    \end{subfigure}
    \begin{subfigure}[b]{0.3\textwidth}
        \includegraphics[width=\textwidth]{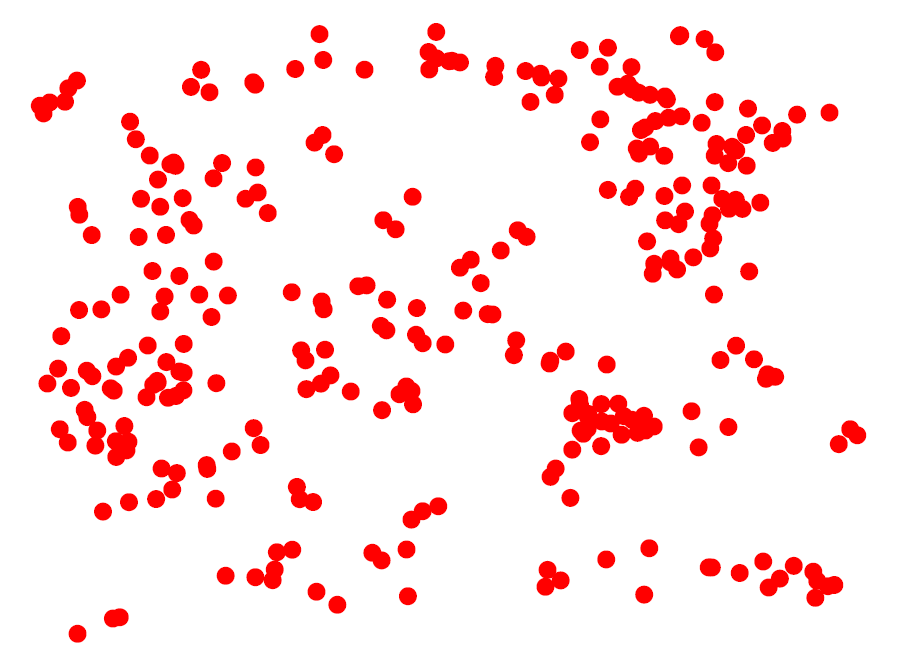}
        \caption{}
    \end{subfigure}
    \begin{subfigure}[b]{0.3\textwidth}
        \includegraphics[width=\textwidth]{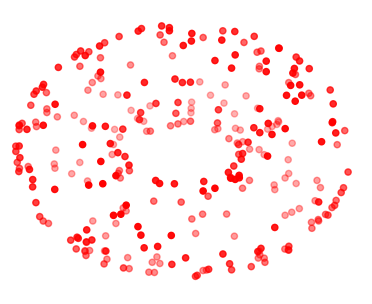}
        \caption{}
    \end{subfigure}
    \caption{Realization of the processes (a), (b) and (c) described in Section \ref{sec:num}.}\label{fig:datasets}
\end{figure}

\begin{figure}
    \centering
    \begin{subfigure}[b]{0.3\textwidth}
        \includegraphics[width=\textwidth]{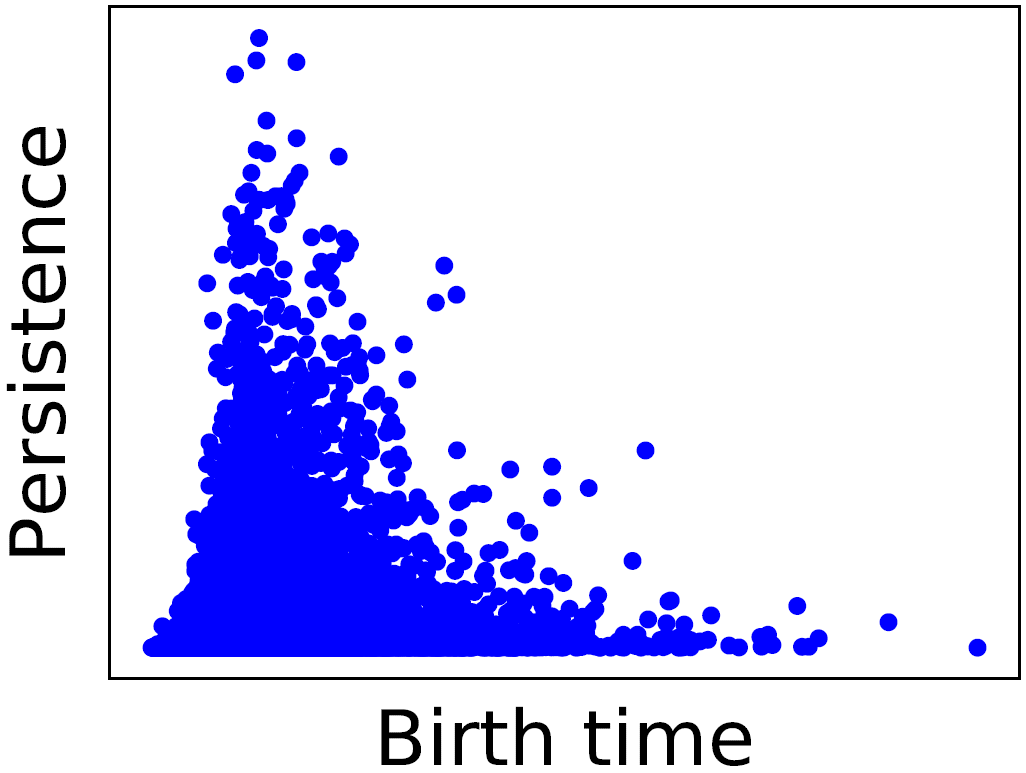}
        \caption{}
    \end{subfigure}
    \begin{subfigure}[b]{0.3\textwidth}
        \includegraphics[width=\textwidth]{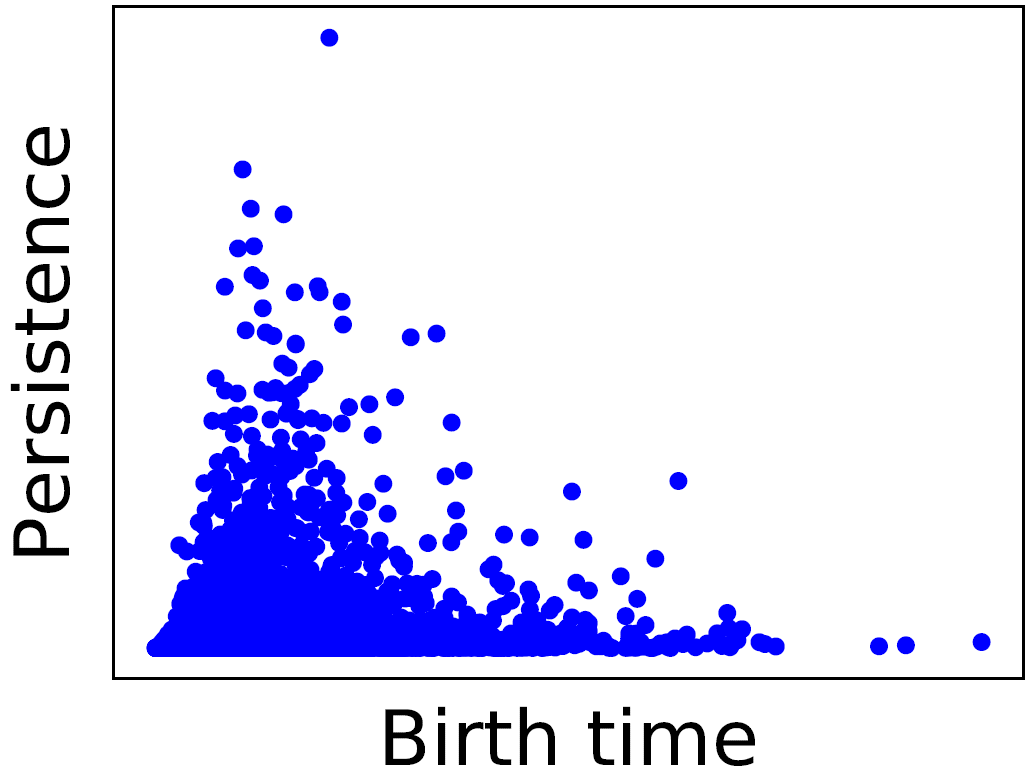}
        \caption{}
    \end{subfigure}
    \begin{subfigure}[b]{0.3\textwidth}
        \includegraphics[width=\textwidth]{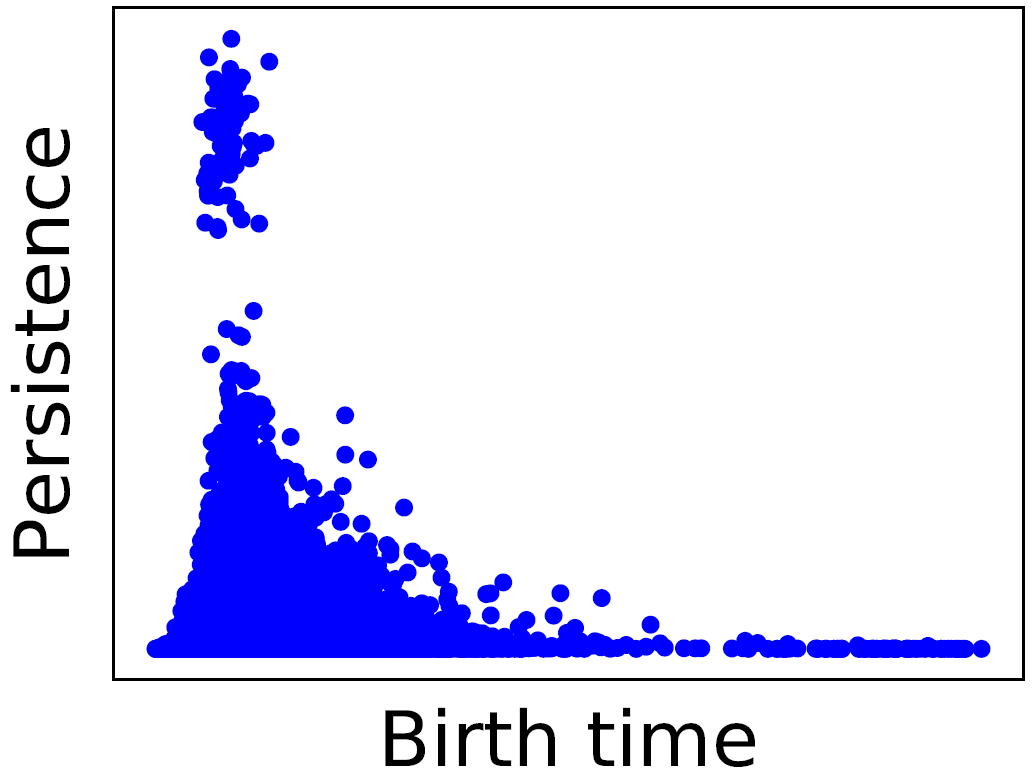}
        \caption{}
    \end{subfigure}
    \caption{Superposition of the $N=40$ diagrams of class (a), (b) and (c), transformed under the map $\rD \to (r_1,r_2-r_1)$.}\label{fig:diagrams}
\end{figure}
\begin{figure}
    \centering
    \begin{subfigure}[b]{0.3\textwidth}
        \includegraphics[width=\textwidth]{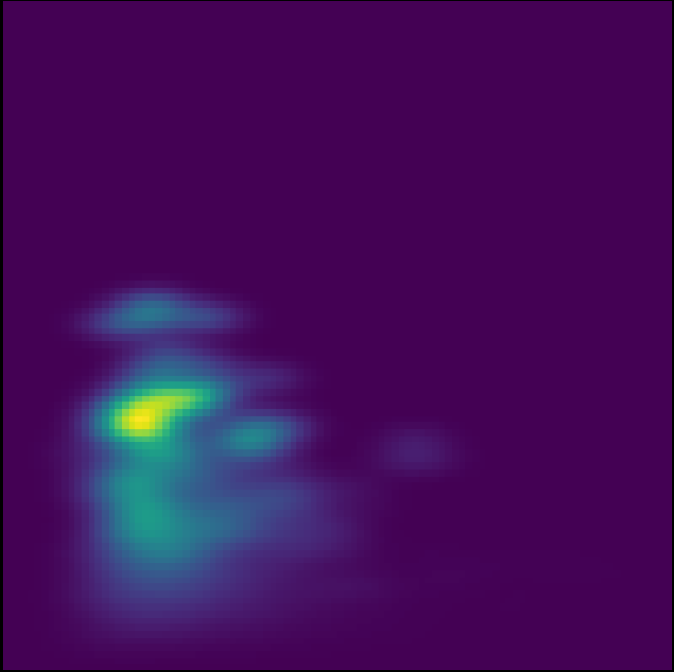}
        \caption{}
    \end{subfigure}
    \begin{subfigure}[b]{0.3\textwidth}
        \includegraphics[width=\textwidth]{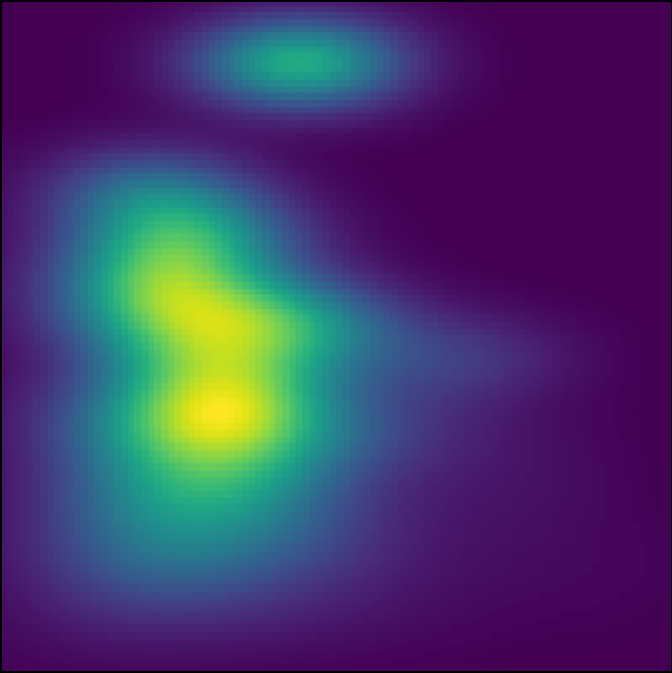}
        \caption{}
    \end{subfigure}
    \begin{subfigure}[b]{0.3\textwidth}
        \includegraphics[width=\textwidth]{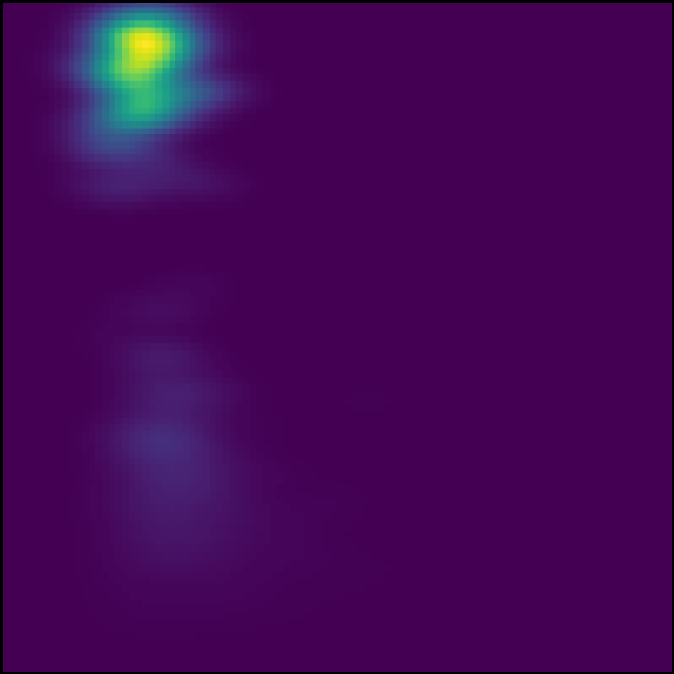}
        \caption{}
    \end{subfigure}
    \caption{Persistence surfaces for each class (a), (b) and (c), computed with the weight function $w(\rD)=(r_2-r_1)^3$ and with the bandwidth matrix selected by the cross-validation procedure.}\label{fig:surfaces}
\end{figure}
For each class of dataset, the score $\hat{J}(H)$ is computed for a set of bandwidth matrices of the form $h^2 \times \begin{bmatrix}
   1 & 0\\
   0 & 1
\end{bmatrix}$, for $50$ values $h$ evenly spaced on a log-scale between $10^{-5}$ and $1$. Note that the computation of $\hat{J}(H)$ only involves the computations of $K_H(\rD_1-\rD_2)$ for points $\rD_1$, $\rD_2$ in different diagrams. Hence, the complexity of the computation of $\hat{J}(H)$ is in $O(T^2)$, where $T$ is the sum of the number of points in the diagrams of a given class. If this is too costly, one may use a subsampling approach to estimate the integrals. The selected bandwidth were respectively $h=0.22, 0.60, 0.17$. Persistence surfaces for the selected bandwidth are displayed in Figure \ref{fig:surfaces}. The persistence of the "true" points of the torus are sufficient to suppress the topological noise: only two yellow areas are seen in the persistence surface of the torus. Note that the two areas can be separated, whereas it is not obvious when looking at the superposition of the diagrams, and would not have been obvious with an arbitrary choice of bandwidth. The bandwidth for class (b) may look to have been chosen too large. However, there is much more variability in class (b) than in the other classes: this phenomenon explains that the density is less peaked around a few selected areas than in class (a).

The cross-validation scheme has also been applied to non-synthetic data: the walk of 3 persons A, B and C, has been recorded using the accelerometer sensor of a smartphone
in their pocket, giving rise to 3 multivariate time series in $\R^3$. Using a sliding window, each series has been splited in a list of 10 times series made of 200 consecutive points. Using a time-delay embedding technique, those new time series are embedded into $\R^9$: these are the point clouds on which we build the Rips filtration. For each person, the set of 10 persistence diagrams is transformed under the map $(r_1,r_2)\mapsto (r_1,r_2-r_1)$. The persistence diagrams are weighted by the weight function $w(\rD)=(r_2-r_1)^3$. For each person, the scores $\hat{J}(H)$ are computed for a set of bandwidth matrix of the form $h^2 \times \begin{bmatrix}
   1 & 0\\
   0 & 1
\end{bmatrix}$, for $20$ values $h$ evenly spaced on a log-scale between $10^{-3}$ and $10^{-1}$. The selected bandwidths are $0.0089, 0.01833$ and $0.0089$ and the corresponding persistence images are displayed in Figure \ref{fig:surfaces_acc}. The three images show very distinct patterns: a reasonable machine learning algorithm will easily make the distinction between the three classes using the images as input.

\begin{figure}
    \centering
    \begin{subfigure}[b]{0.3\textwidth}
        \includegraphics[width=\textwidth]{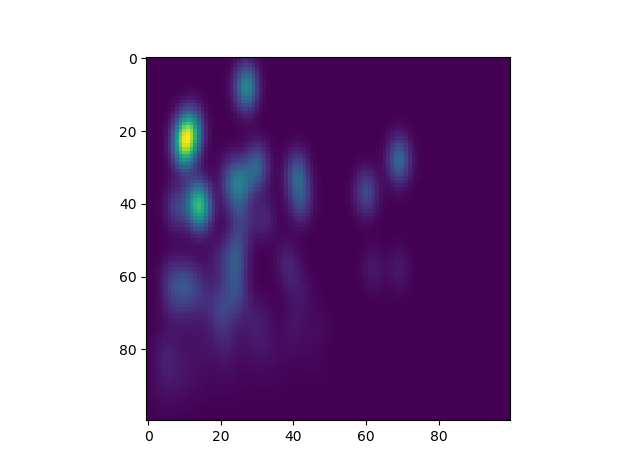}
        \caption{}
    \end{subfigure}
    \begin{subfigure}[b]{0.3\textwidth}
        \includegraphics[width=\textwidth]{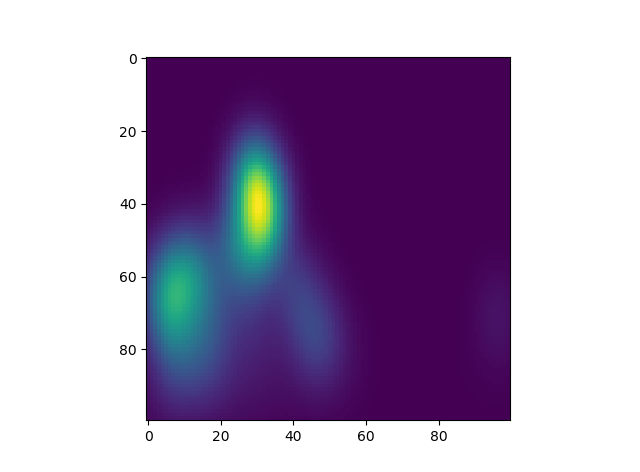}
        \caption{}
    \end{subfigure}
    \begin{subfigure}[b]{0.3\textwidth}
        \includegraphics[width=\textwidth]{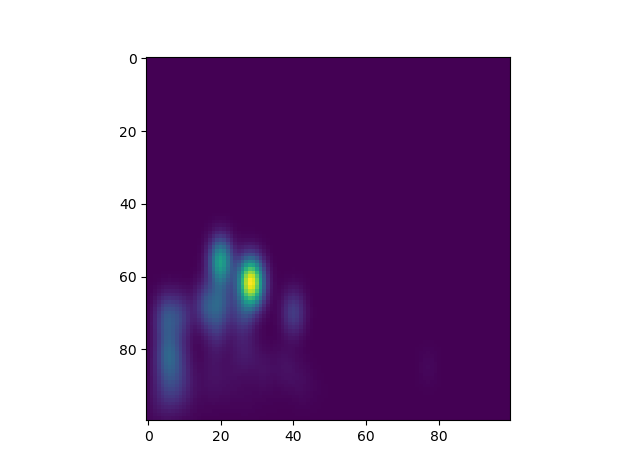}
        \caption{}
    \end{subfigure}
    \caption{Persistence surfaces for each person A,B and C, computed with the weight function $w(\rD)=(r_2-r_1)^3$ and with the bandwidth matrix selected by the cross-validation procedure.}\label{fig:surfaces_acc}
\end{figure}

\section{Conclusion and further works}
Taking a measure point of view to represent persistence diagrams, we have shown that the expected behavior of persistence diagrams built on top of random point sets reveals to have a simple and interesting structure: a measure on $\R^2$ with density with respect to Lebesgue measure that is as smooth as the random process generating the data points! This opens the door to the use of effective kernel density estimation techniques for the estimation of the expectation of topological features of data. Our approach and results also seem to be particularly well-suited to the use of recent results on the Lepski method for parameter selection \cite{lacour2016estimator} in statistics, a research direction that deserves further exploration. 
As many persistence-based features considered among the literature - persistence images, birth and death distributions, Betti curves,... - can be expressed as linear functional of the discrete measure representation of diagrams, our results immediately extend to them. The ability to select the parameters on which these features are dependent in a well-founded statistical way also opens the door to a well-justified usage of persistence-based features in further supervised and un-supervised learning tasks.

\bibliographystyle{alpha}
\bibliography{paper.bib}

\appendix
\section{Proofs of the subanalytic elementary lemmas}

\firstProp*
\begin{proof}
\begin{enumerate}
\item[(i)] Section I.2.1 in \cite{shiota1997geometry} states that $A(f)$ is subanalytic. Therefore, its complement $E$ is also subanalytic: it is enough to show that $E$ is of empty interior to conclude. 
\begin{claim} The set $F$ of points $x$ where $f$ is not analytic but $G_f$ is locally a real analytic manifold in $(x,f(x))$ is a subanalytic set of empty interior.
\end{claim}
\begin{claimproof} Assume $F$ contains an open set $U$. Replacing $U$ by a smaller open set if necessary, there exists some local parametrization of $U_f = \{(x,f(x)),\ x\in U\}$ by some analytic function $\Phi : V \to \R$, $V$ being a neighborhood of $U_f$ in $M\times \R$. Denote by $\nabla^u \Phi \in \R$ the gradient of $\Phi$ with respect to the real variable $u\in \R$. The set $Z$ on which $\nabla^u \Phi = 0$ is an analytic subset of $V$. As $G_f$ is the graph of a function, $Z \cap G_f$ is made of isolated points: one can always assume that those points are not in $U_f$. Therefore, there exists some neighborhood $V'$ of $U_f$ which does not intersect $Z$. One can now apply the analytic implicit function theorem (see for instance \cite[Section 8]{kaup1983holomorphic}) anywhere on $U_f$: for $(x_0,u_0)\in U_f$, there exists some neighborhood $W \subset V'$ and an analytic function $g:\Omega\to \R$, $\Omega$ being a neighborhood of $x_0$, such that, on $W$
\[ \Phi(x,u) = 0 \Longleftrightarrow u=g(x).\]
As we also have $\Phi(x,u)=0$ if and only if $u=f(x)$, $f\equiv g$ on $\Omega$ and $f$ is analytic on $\Omega$. This is a contradiction with having $f$ not analytic in every point of $U$.
\end{claimproof}

Now, the set $E$ is the union of $F$ and of $E\cap G$ where $G$ is the projection on $M$ of $\mathrm{Sing}(G_f)$. As, by definition, $\mathrm{Sing}(G_f)$ is of empty interior, $G$ is also of empty interior. Therefore, $E$ is of empty interior, which is equivalent to say that its dimension is smaller than $d$.
\item[(i)] See \cite[Section II.1.1]{shiota1997geometry}.
\item[(ii)] See \cite[Section II.1.6]{shiota1997geometry}.
\end{enumerate}
\end{proof}
\usefulLem*
\begin{proof} Write $k$ the dimension of $X$. First, one can always assume that $X$ is closed, as $\mathcal{H}_d(\overline{X}) \geq \mathcal{H}_d(X)$. Therefore, there exists some real analytic manifold $N$ of dimension $k$ and a proper real analytic mapping $\Psi : N \to M$ such that $\Psi(N)=X$ (see \cite[Section I.2.1]{shiota1997geometry}). The set $X$ can be written as the union of some compact sets $X_K$ for $K\geq 0$. It is enough to show that $\mathcal{H}_d(X_K) = 0$. The set $X_K$ can be written $\Psi(\Psi^{-1}(X_K))$, where $\Psi^{-1}(X_K)$ is some compact subset of $N$. We have $\mathcal{H}_d(\Psi^{-1}(X_K)) = 0$ because $N$ is of dimension $k<d$. Furthermore, as $\Psi$ is analytic on $Y$, it is Lipschitz on $\Psi^{-1}(X_K)$. Therefore, $\mathcal{H}_d(\Psi(\Psi^{-1}(X_K)))=\mathcal{H}_d(X_K)$ is also null.
\end{proof}

\section{Proof of Theorem \ref{thm:main_thm'}}
\mainThmbisState*
We indicate how to change the proof of Theorem \ref{thm:main_thm} when assumption (K5') is satisfied instead of assumption (K5). In the partition $E_1(x),\dots,E_L(x)$ of $\mathcal{F}_n$, the set $E_1(x)$ plays a special role: it corresponds to the value $r_1=0$ and contains all the singletons, which satisfy $\varphi[\{j\}] \equiv 0$ by assumption. Lemma \ref{lemma1} holds for $l>1$ and one can always define $J_1 = \{1\}$ to be a minimal element of $E_1(x)$. With this convention in mind, it is straightforward to check that Lemma \ref{lemma2} still holds and that Lemma \ref{lem:grad_non_null} is satisfied as well for $l>1$. Now, one can define in a likewise manner the sets $V_r$. For $x\in V_r$, the diagram $D_s[\KK(x)]$ is still decomposed $\sum_{i=1}^N \delta_{\rD_i}$, with $\rD_i = (\varphi[J_{l_1}](x),\varphi[J_{l_2}](x))$. If $s>0$, the end of the proof is similar. However, for $s=0$, the pairs of simplices $(J_{l_1},J_{l_2})$ are made of one singleton $J_{l_1}$ and of one 2-simplex $J_{l_2}$. As $\varphi$ is null on singletons, the points in this diagram are all included in the vertical line $L_0 \defeq \{0\} \times [0,\infty)$. The map $\Phi_{ir}:x\in V_r\mapsto \rD_i \in L_0$ has a differential of rank 1, as Lemma \ref{lem:grad_non_null} ensures that $\nabla^j \varphi[J_{l_2}](x) \neq 0$ for $j \in J_{l_2}$. One can apply the coarea formula to $\Phi_{ir}$ to conclude to the existence of a density with respect to the Lebesgue measure on $L_0$.

\section{Proof of Corollary \ref{cor:main_cor}}
\mainCorState*
The diagram $D_s[\KK(\X)]$ can be written
\begin{equation}
D_s[\KK(\X)] = \sum_{n\geq 0} \mathbbm{1}\{|\X|=n\} D_s[\KK(\X)],
\end{equation}
and Theorem \ref{thm:main_thm} states that $\mathbbm{1}\{|\X|=n\} D_s[\KK(\X)]$ has a density $p_n$ with respect to the Lebesgue measure on $\Delta$. Take $B$ a Borel set in $\Delta$:
\begin{align*}
E[D_s[\KK(\X)]](B) &= \sum_{n\geq 0} E[\mathbbm{1}\{|\X|=n\} D_s[\KK(\X)]](B) \\
&= \sum_{n\geq 0} \int_B p_n = \int_B \sum_{n\geq 0} p_n \mbox{ by Fubini-Torelli's theorem.} 
\end{align*}
It is possible to use Fubini-Torelli's theorem because $E[D_s[\KK(\X)]](B)$ is finite. Indeed, as $D_s[\X]$ is always made of less than $2^{|\X|}$ points, and as we have supposed that $E\left[ 2^{|\X|}\right]<\infty$, the measure $E[ D_s[\KK(\X)]]$ is finite as well.

\section{Proof of Theorem \ref{thm:smoothness}}
\smoothnessState*
Given the expression \eqref{densityLoc}, it is sufficient to show that integrating a function along the fibers is a smooth operation in the fibers. We only show that the density is continuous. Continuity of the higher orders derivatives is obtained in a similar fashion. The proof is a standard application of the implicit function theorem.

Using the same notations than in the proof of Theorem \ref{thm:main_thm}, fix $1\leq r \leq R$ and $1\leq i \leq N_r$. We will show that $p_{ir}$ is continuous. As the indices $r$ and $i$ are now fixed, we drop the dependency in the notation: $V \defeq V_r$ and $\Phi \defeq \Phi_{ir}$. By using a partition of unity and taking local diffeomorphisms, one can always assume that $V \subset \R^{nd}$. Define the function $f: (x,u)\in V\times\Delta \mapsto \Phi(x)-u \in \R^2$. We have already shown in the proof of Theorem \ref{thm:main_thm} that for $x_0\in V$, there exists two indices $a_1$ and $a_2$ (depending on $x_0$) such that the minor $M(x_0)=(D \Phi(x_0))_{a_{1,2}} $ is invertible. Rewrite $x \in V$ in $(y,z)$ where $z = (x_{a_1},x_{a_2})\in \R^2$. By the implicit function theorem, for $(x_0,u_0)$ such that $f(x_0,u_0)=0$, there exists a neighborhood $\Omega_{x_0} \subset V\times \Delta$ of $(x_0,u_0)$ and an analytic function $g_{x_0}: W_{y_0}\times Y_{u_0} \to \R^2$ defined on a neighborhood of $(y_0,u_0)$ such that for $(x,u)\in \Omega_{x_0}$
\[ f(x,u) = 0 \Longleftrightarrow z=g_{x_0}(y,u).\]
The sets $(\Omega_{x_0})_{x_0 \in V}$ constitutes an open cover of the fiber $f^{-1}(0)$. Consider a smooth partition of unity $(\rho_{x_0})_{x_0 \in V}$ subordinate to this cover. Then, for all $(x,u) \in f^{-1}(0)$
\begin{align*}
(J\Phi(x))^{-1}\kappa(x) &= \sum_{x_0 \in V} \rho_{x_0}(y,u,g_{x_0}(y,u)) (J\Phi(y,g_{x_0}(y,u)))^{-1}\kappa(y,g_{x_0}(y,u))
\end{align*}
Therefore,
\begin{align}
p_{ir}(u)&=\int_{x\in \Phi^{-1}(u)} (J\Phi(x))^{-1}\kappa(x) d\mathcal{H}_{nd-2}(x)  \nonumber \\
&=\sum_{x_0 \in V} \int_{y \in W_{y_0}}\rho_{x_0}(y,u,g_{x_0}(y,u)) (J\Phi(y,g_{x_0}(y,u)))^{-1}\kappa(y,g_{x_0}(y,u))dy. \label{sumPart}
\end{align}

We are now faced with a classical continuity under the integral sign problem. First, the Cauchy-Binet formula (see \cite[Example 2.15]{kwak2004linear}) states that $J\Phi$ is equal to the square root of the sum of the squares of the determinants of all $2\times 2$ minors of $D\Phi$. Therefore, $J\Phi(x)$ is larger than the determinant of $M(x)$, the minor of $f$ of indices $a_1$ and $a_2$. The implicit function theorem gives the exact value of $M(x)$. Indeed, for $X=(x,u)\in \Omega_{x_0}$, and for any index $k$,
\begin{equation}
\frac{\partial g}{\partial X_k}(y,u)= -\left(M^{-1} \cdot \frac{\partial f}{\partial X_k}\right)(y,u,g(y,u))
\end{equation}
Take $X_k=u_{1,2}$. Then, $\partial f/\partial X_k = (-1,0)$, resp. $(0,-1)$. Therefore,
\begin{equation}\label{boundJac}
 M^{-1}(y,u,g(y,u)) = \frac{\partial g}{\partial u}(y,u,g(y,u))
\end{equation}

As $\rho_{x_0}$ has a compact support, it suffices to show that the integrand is bounded by a constant independent of $u$. The only issue is that $(J\Phi)^{-1}$ may diverge. Equation \eqref{boundJac} shows that it is bounded by $\det \partial g/\partial u$. This is bounded, as $g$ is analytic on the compact support of $\rho_{x_0}$: each term in the sum \eqref{sumPart} is continuous. By the compactness of $M$ and $f^{-1}(0)$, all the partitions of unity can be taken finite, and a finite sum of continuous functions is continuous. This proves the continuity of $p$.

\section{Proof of Corollary \ref{cor:betti}}
\corBettiState*
Define $f(r,u)$ to be equal to $1$ if $u_1\leq r \leq u_2$ and $0$ otherwise. Then, $\beta^r_s(\KK(\X))$ is equal to $D_s[\KK(\X)](f(r,\cdot))$. Therefore, the expectation $E[\beta^r_s(\KK(\X))]$ is equal to
\begin{equation}
\int p(u)f(r,u)du.
\end{equation}
As we assumed that the hypothesis of Theorem \ref{thm:smoothness} were satisfied, the density $p$ is smooth. Moreover, $p(u)f(r,u)$ is smaller than $p(u)$. The function $p$ being integrable, one can apply the continuity under the integral sign theorem to conclude that $r\mapsto E[\beta^r_s(\KK(\X))]$ is continuous. Higher-order derivatives are obtained in a similar fashion.

\end{document}